\documentclass[10pt,conference]{ieeeconf}

\usepackage{mathrsfs}
\usepackage{graphics} 
\usepackage{amsmath} 
\usepackage{amssymb}  

\usepackage{amsthm}
\usepackage{cite}
\usepackage{bm}
\usepackage{acronym}
\usepackage{paralist}
\usepackage{color}
\usepackage{tikz}
\usepackage{hyperref}
\usepackage{graphicx}
\usepackage{soul}
\usepackage{mathtools}
\usepackage{siunitx}
\usepackage{empheq}
\usepackage{accents}
\usepackage{mdframed}
\renewcommand{\theenumi}{\roman{enumi}}

\makeatletter
\hypersetup{colorlinks=true}
\AtBeginDocument{\@ifpackageloaded{hyperref}
  {\def\@linkcolor{blue}
  \def\@anchorcolor{red}
  \def\@citecolor{red}
  \def\@filecolor{red}
  \def\@urlcolor{red}
  \def\@menucolor{red}
  \def\@pagecolor{red}
\begingroup
  \@makeother\`%
  \@makeother\=%
  \edef\x{%
    \edef\noexpand\x{%
      \endgroup
      \noexpand\toks@{%
        \catcode 96=\noexpand\the\catcode`\noexpand\`\relax
        \catcode 61=\noexpand\the\catcode`\noexpand\=\relax
      }%
    }%
    \noexpand\x
  }%
\x
\@makeother\`
\@makeother\=
}{}}
\makeatother

\newtheorem{Theorem}{Theorem}
\newtheorem{Lemma}{Lemma}
\newtheorem{Problem}{Problem}
\newtheorem{Remark}{Remark}

\newtheorem{Example}{Example}

\newtheorem{Definition}{Definition}

\DeclareMathOperator{\R}{\mathbb R}

\DeclareMathOperator*{\argmin}{arg\,min}

\newcommand{\bequ}{\begin{eqnarray}}
\newcommand{\eequ}{\end{eqnarray}}

\newcommand{\bb}{\boldsymbol}

\IEEEoverridecommandlockouts

\def\BibTeX{{\rm B\kern-.05em{\sc i\kern-.025em b}\kern-.08em
    T\kern-.1667em\lower.7ex\hbox{E}\kern-.125emX}}

\begin{document}

\title{\LARGE \bf Safety Under Uncertainty: Tight Bounds with Risk-Aware Control Barrier Functions}

\author{Mitchell Black$^1$ \and Georgios Fainekos$^2$ \and Bardh Hoxha$^2$ \and Danil Prokhorov$^2$ \and Dimitra Panagou$^3$
\thanks{$^1$Dept. of Aerospace Engineering,  Univ. of Michigan, 1320 Beal Ave, Ann Arbor, MI 48109, USA; \texttt{\{mblackjr\}@umich.edu}.}
\thanks{$^2$Toyota North America Research \& Development, 1555 Woodridge Ave, Ann Arbor, MI 48105, USA; \texttt{\{georgios.fainekos, bardh.hoxha, danil.prokhorov\}@toyota.com}.}
\thanks{$^3$Dept. of Robotics and Dept. of Aerospace Engineering,  Univ. of Michigan, Ann Arbor, MI 48109, USA; \texttt{\{dpanagou\}@umich.edu}.}
}
\maketitle  


\begin{abstract}\label{sec: abstract}
We propose a novel class of risk-aware control barrier functions (RA-CBFs) for the control of stochastic safety-critical systems. Leveraging a result from the stochastic level-crossing literature, we deviate from the martingale theory that is currently used in stochastic CBF techniques and prove that a RA-CBF based control synthesis confers a tighter upper bound on the probability of the system becoming unsafe within a finite time interval than existing approaches. We highlight the advantages of our proposed approach over the state-of-the-art via a comparative study on an mobile-robot example, and further demonstrate its viability on an autonomous vehicle highway merging problem in dense traffic. 
\end{abstract}


\section{Introduction}\label{sec.intro}

A safe autonomous system is the gate to a fully autonomous system. Since the arrival of control barrier functions (CBFs) as a tool for safe control design and system verification \cite{Ames2017CBFs, Chen2021Guaranteed,Cortez2021,Garg2021Robust,Yaghoubi2021RiskBounded}, the field of safety-critical systems has drawn nearer to passing through this gate. Intuitively, a valid CBF constitutes a certificate that all system trajectories beginning within the set of safe states shall remain safe for all future time. For a class of control-affine dynamical systems, it is popular to included CBFs as linear constraints in a quadratic program (QP) based control law. And while the potential of this set-theoretic approach to render, preserve, and/or verify safety in a system has been demonstrated in applications like autonomous driving \cite{Son2019Nonaffine,Black2022ffcbf,Yaghoubi2021RiskBounded}, quadrotor control \cite{Wang2018SafeLearning,Black2022Fixed}, and multi-agent systems \cite{Chen2021Guaranteed, Cheng2020SafeMultiagent}, its guarantees of safety may be lost in the absence of a complete system model.

In the deterministic setting, various works have addressed this problem by modifying standard CBF conditions. For example, robust-CBF approaches were proposed for safe control design under bounded perturbations to the system dynamics \cite{Jankovic2018Robust, Garg2021Robust,Yaghoubi2020Training} or measurement error \cite{Cosner2021Measurement,Dean2021Guaranteeing}, though the worst-case assumptions beget conservatism. Under linearly parametric model uncertainty, adaptive-CBF approaches have been used to both learn \cite{Black2022Fixed} and compensate for the effect of \cite{Taylor2019aCBF} unknown parameters in the system dynamics. For more general nonlinear uncertainty, Gaussian processes have been used to learn non-parametric, probabilistic models of the system \cite{Jagtap2020Control,Cheng2020SafeMultiagent,Castaneda2021Pointwise}, residual terms in the CBF \cite{Khan2021Safety}, and barrier functions directly \cite{Wang2018SafeLearning}, among others. Additional works have adopted chance-constrained CBF conditions for probabilistic models \cite{Khojasteh2020Probabilistic,Lyu2021Probabilistic,luo2020multi}, but do not consider safety over a time interval.

In many practical applications, the system behavior instead may be modeled by a class of stochastic differential equations (SDEs). Beginning with \cite{Prajna2007Stochastic} and the stochastic barrier certificate, CBF development in the stochastic setting has leaned heavily on martingale theory for both discrete- and continuous-time stochastic processes. Stochastic CBFs (S-CBFs), introduced in \cite{Santoyo2021FiniteStochastic} and adapted for risk-bounded control in \cite{Yaghoubi2021RiskBounded,Yaghoubi2021RiskKalman}, leverage martingales to bound the probability that a system becomes unsafe over a finite time interval. While useful in theory, in practice the probability of safety is severely limited by the initial condition. In \cite{Clark2021Stochastic} this problem is addressed via reciprocal and zeroing CBFs for stochastic systems, with claims of safety with probability one, though the required level of conservatism is unclear.

In this paper, we deviate from martingale theory and make the following contributions:
\begin{itemize}
    \item We introduce a new class of risk-aware control barrier functions (RA-CBFs) that uses a generator\footnote{The (infinitesimal) generator of a stochastic process is analogous to the Lie derivative for deterministic systems.} condition derived from the stochastic level-crossing literature to obtain an upper bound on the probability that the system becomes unsafe over a finite time interval.
    \item We derive conditions under which our RA-CBF controller guarantees a smaller upper bound on the risk of the system becoming unsafe than existing S-CBF methods, and further show via a 100,000 trial numerical study that our controller results in less conservative behavior despite this stronger guaranteed risk protection.
    \item We consider an autonomous vehicle highway merging problem and demonstrate the efficacy of our proposed RA-CBF based controller in successfully merging amongst dense traffic under a required safety probability of $99\%$.
\end{itemize}

\section{Preliminaries and Problem Formulation}\label{sec.prelims}
The uniform distribution supported by $a$ and $b$ is $U[a,b]$. A bolded $\bb{x}_t$ denotes a vector stochastic process at time $t$.
The Gauss error function is $\textrm{erf}(z) = \frac{2}{\sqrt{\pi}}\int_0^ze^{-t^2}dt$, and $\textrm{erf}^{-1}(\cdot)$ is its inverse. The trace of a matrix $\bb{M} \in \R^{n \times n}$ is $\textrm{Tr}(\bb{M})$. The Lie derivative of a function $\phi:\mathbb R^n\rightarrow \mathbb R$ along a vector field $f:\mathbb R^n\rightarrow\mathbb R^n$ at a point $x\in \mathbb R^n$ is $L_f\phi(x) \triangleq \frac{\partial \phi}{\partial x} f(x)$.

\subsection{Preliminaries}
We consider the following class of nonlinear, control-affine, stochastic differential equations (SDE),
\begin{equation}\label{eq.stochastic_system}
    d\bb{x}_t = \big(f(\bb{x}_t) + g(\bb{x}_t)\bb{u}_t\big)dt + \sigma(\bb{x}_t)d\bb{w}_t, 
\end{equation}
where $\bb{x} \in \mathcal{X} \subseteq \R^n$ denotes the state, $\bb{u} \in \mathcal{U} \subseteq \R^m$ the control input, and $\bb{w} \in \R^q$ a standard $q$-dimensional Wiener process (i.e., Brownian motion) defined over the complete probability space $(\Omega, \mathcal{F}, P)$ for sample space $\Omega$, $\sigma$-algebra $\mathcal{F}$ over $\Omega$, and probability measure $P: \mathcal{F} \rightarrow [0,1]$. We consider a class of memoryless, state-feedback controllers such that the control signal is $\bb{u}_t = k(\bb{x}_t)$,
with $f: \mathcal{X} \rightarrow \R^n$, $g: \mathcal{X} \rightarrow \R^{n \times m}$, and $k: \mathcal{X} \rightarrow \mathcal{U}$ known, locally Lipschitz, and bounded on $\mathcal{X}$, which is assumed to be bounded. We consider that $\sigma: \R^n \rightarrow \R^{n \times q}$ from \eqref{eq.stochastic_system} also satisfies these regularity conditions, and thus assume that for all $\bb{u} \in \mathcal{U}$ and $\bb{x}_0 \in \mathcal{X}_0 \subset \R^n$ the process $\{\bb{x}_t: t \in [0, \infty)\}$ is a strong solution to \eqref{eq.stochastic_system} (see \cite[Ch. 5, Def. 2.1]{Karatzas1998Brownian}).
For strong solutions, the generator 
is defined as follows.
\begin{Definition}\hspace{-0.2pt}\cite[Def. 7.3.1]{Oksendal2003Stochastic}\label{def.generator}
    The (infinitesimal) generator $\mathcal{A}$ of $\bb{x}_t$ is defined by
    \begin{equation}\label{eq.generator}
        \mathcal{A}\phi(\bb{y}) = \lim_{t \downarrow 0}\frac{\mathbb{E}\left[\phi(\bb{x}_t) \; | \; \bb{x}_0 = \bb{y}\right] - \phi(\bb{y})}{t},
    \end{equation}
    where $\phi: \R^n \rightarrow \R$ belongs to  $\mathcal{D}_\mathcal{A}$, the set of all functions such that the limit exists for all $\bb{x} \in \R^n$.
\end{Definition}
The generator is the stochastic analog to the Lie derivative for deterministic systems in that it characterizes the derivative of $\phi$ over the trajectories of \eqref{eq.stochastic_system} in expectation. By \cite[Thm. 7.3.3]{Oksendal2003Stochastic}, for a twice continuously differentiable function $\phi$ with compact support, i.e., $\phi \in \mathcal{C}_0^2(\R^n) \subset \mathcal{D}_\mathcal{A}$, the generator $\mathcal{A}$ of $\bb{x}_t$ is described by
\begin{equation}
    \mathcal{A}\phi(\bb{x}) = 
    L_f\phi(\bb{x}) + L_g\phi(\bb{x})k(\bb{x}) +
    \frac{1}{2}\textrm{Tr}\left(\sigma(\bb{x})^T\frac{\partial^2 \phi}{\partial \bb{x}^2}\sigma(\bb{x})\right), \nonumber
\end{equation}
which we denote $\Gamma_\phi(\bb{x},\bb{u}) \coloneqq \mathcal{A}\phi(\bb{x})$ by using $\bb{u} = k(\bb{x})$.

Consider the set $S$ defined by a twice continuously differentiable, positive semi-definite function $B: \R^n \rightarrow \R$:
\begin{equation}\label{eq.safe_set}
    S = \{\bb{x} \in \R^n: \; 0 \leq B(\bb{x}) < 1\},
\end{equation}
and assume that is also known that for some $\gamma \in [0, 1]$,
\begin{equation}\label{eq.B0_leq_gamma}
    B(\bb{x}) \leq \gamma, \; \forall \bb{x} \in \mathcal{X}_0.
\end{equation}
In the deterministic setting, the set $S$ is said to be \textit{forward-invariant} if $\bb{x}_0 \in S \implies \bb{x}_t \in S$, $\forall t \geq 0$. In this paper, we assume that $S$ denotes the set of safe states for \eqref{eq.stochastic_system} and therefore use the notions of forward invariance and safety interchangeably. In the stochastic setting, however, there may be failure cases in which $\bb{x}_t$ exits $S$, i.e., the system becomes unsafe. We therefore consider the stopped process, $\Tilde{\bb{x}}_t$, and probabilistic forward invariance, adapted from \cite{Kofman2012Probabilistic}.
\begin{Definition}\hspace{-0.3pt}\cite{Kushner1967Stochastic}
    Suppose that $\tau>0$ is the first time of exit of $\bb{x}_t$ from the open set $S$. The stopped process $\Tilde{\bb{x}}_t$ is
    \begin{equation}
        \Tilde{\bb{x}}_t = \begin{cases}
        \bb{x}_t; & t < \tau, \\
        \bb{x}_\tau; & t \geq \tau.
        \end{cases}\nonumber
    \end{equation} 
\end{Definition}
\begin{Definition}
    Let $0 < p \leq 1$, and consider the stopped process over an interval of length $T>0$, i.e., $\{\Tilde{\bb{x}}_t: t \in [0,T]\}$, w.r.t. the set $S$ defined by \eqref{eq.safe_set}. The set $S \subset \mathcal{X} \subseteq \R^n$ is a \textbf{probabilistic forward-invariant set with probability p} for system \eqref{eq.stochastic_system} over the interval $[0,T]$ if $P\left(\Tilde{\bb{x}}_t \in S, \forall t \in [0,T]\right) \geq p$.
\end{Definition}
Thus, $S$ is \textbf{safe} with probability $p$ over the interval $[0,T]$ if it is a probabilistic forward-invariant set with probability $p$ over $[0,T]$. Alternatively, the probability $\rho$ of the system becoming unsafe over $[0,T]$, i.e., $\rho \coloneqq P(\exists t \in [0,T]: \Tilde{\bb{x}}_t \notin S)$, is bounded by $\rho \leq 1 - p$. We refer to $\rho$ in the remainder as the "system risk". One approach to bounding the system risk of \eqref{eq.stochastic_system} is to use S-CBFs in the control design \cite{Santoyo2021FiniteStochastic,Yaghoubi2021RiskBounded}. 
\begin{Definition}\label{def.stochastic_cbf}
    Consider a set $S \subset \R^n$ defined by \eqref{eq.safe_set} for a twice continuously differentiable, positive semi-definite function $B$ satisfying \eqref{eq.B0_leq_gamma}. The function $B$ is a \textbf{stochastic control barrier function} (S-CBF) defined on the set $S$ if there exist $\alpha,\beta \geq 0$ such that for the system \eqref{eq.stochastic_system} the generator $\Gamma_B(\bb{x},\bb{u})$ satisfies the following condition, for all $\bb{x} \in S$,
    \begin{equation}
        \begin{aligned}\label{eq.scbf_condition}
            \inf_{\bb{u} \in \mathcal{U}}\Gamma_B(\bb{x},\bb{u}) \leq -\alpha B(&\bb{x}) + \beta.
        \end{aligned}
    \end{equation}
\end{Definition}
A valid S-CBF guarantees that the system risk is bounded from above, as shown in the following \cite[Prop. 1]{Santoyo2021FiniteStochastic}.
\begin{Theorem}\label{thm.stochastic_cbf_bounds}
    Consider a stochastic system of the form \eqref{eq.stochastic_system}, a set of safe states $S$ implicitly defined by a function $B$ as in \eqref{eq.safe_set}, and the interval $[0, T]$ for $T>0$. Let the probability that $\{\Tilde{\bb{x}}_t: t \in [0,T]\}$ exits $S$ be denoted $\rho_{S\text{-}CBF} \coloneqq P(\exists t\in [0,T]: \Tilde{\bb{x}}_t \notin S \; | \; \Tilde{\bb{x}}(0) \in \mathcal{X}_0)$. If $B$ is a stochastic control barrier function for \eqref{eq.stochastic_system} over the set $S$, then
    \begin{equation}\label{eq.s_cbf_probability_bounds}
        \rho_{S\text{-}CBF} \leq \begin{cases}
        1 - \left(1 - \gamma\right)e^{-\beta T}; & \alpha > 0 \; \textrm{and} \; \alpha \geq \beta, \\
        \left(\gamma + (e^{\beta T} - 1)\frac{\beta}{\alpha}\right)e^{-\beta T}; & \alpha > 0 \; \textrm{and} \; \alpha < \beta, \\
        \gamma + \beta T; & \alpha=0.
        \end{cases}
    \end{equation}
\end{Theorem}

\begin{Remark}\label{rmk.s_cbf_bounds}
    A S-CBF controller can certify that at best a fraction of $1-\gamma$ of the trajectories will be safe over a time interval for any choice of $\alpha, \beta, T \geq 0$. Note that, due to the martingale origins of S-CBFs, the strength of the process noise ($\sigma(\bb{x})$) in \eqref{eq.stochastic_system} does not appear in \eqref{eq.s_cbf_probability_bounds}.
    This motivates the problem formalized in Section \ref{subsubsec.problem}.    
\end{Remark}

For control design, it has become popular to synthesize a variety of CBFs (e.g. S-CBFs \cite{Yaghoubi2021RiskBounded}, chance-constrained CBFs \cite{Khojasteh2020Probabilistic, Lyu2021Probabilistic},
robust CBFs \cite{Garg2021Robust}, etc.) with a nominal controller via quadratic program (QP) based control laws of the form
\begin{subequations}\label{eq.cbf_qp_controller}
\begin{align}
    \bb{u}^* = \argmin_{\bb{u} \in \mathcal{U}} &\frac{1}{2}\|\bb{u}-\bb{u}_0\|^2 + \frac{1}{2}w\delta^2 \label{subeq.cbf_qp_J}\\
    \nonumber &\textrm{s.t.} \\
    A\bb{u} + b &+ c\delta \leq 0, \label{subeq.cbf_constraint}
\end{align}
\end{subequations}
where $\bb{u}_0$ is the nominal control law, $\delta$ is a slack variable with weight $w \geq 0$, and \eqref{subeq.cbf_constraint} represents a generic CBF constraint, with $b,c \in \R$, $A \in \R^{1 \times m}$. In the remainder, we use \eqref{eq.cbf_qp_controller} to compare emergent behaviors of systems under S-CBFs and our proposed risk-aware CBF.


\subsection{Problem Formulation}\label{subsubsec.problem}
Based on Remark \ref{rmk.s_cbf_bounds}, we hypothesize that S-CBFs may introduce unnecessary conservatism into risk-aware control design. We use an illustrative example to show that this is indeed true, and thereby motivate the problem.

\begin{Example}\label{ex.single_integrator}
Consider a mobile robot seeking to achieve the following objective: visit a circular region of radius $R_g > 0$ centered on $s_g = [x_g \; y_g]^T$, defined with respect to the origin $s_0$ of an inertial frame $\mathcal{I}$, while remaining inside a circular region of radius $R_c$ centered on $s_0$. The goal specification may be thought of as visiting a point of interest, while the constraint may model e.g. limited communication range. We choose $s_g = [2 \; 2]^T$, $R_c=1$, $R_g=0.25$ such that the goal and safe sets do not intersect. 
Assume that the robot may be modeled as a $2$D stochastic single-integrator,
\begin{equation}\label{eq.2d_single_integrator_model}
    d\bb{z}_t = \begin{bmatrix}1 & 0 \\ 0 & 1 \end{bmatrix}\begin{bmatrix}v_x \\ v_y \end{bmatrix}dt + \begin{bmatrix}\sigma_x & 0 \\ 0 & \sigma_y \end{bmatrix}d\bb{w}_t,
\end{equation}
where $\bb{z} = [x \; y]^T$ denotes the robot's position (in m) with respect to $s_0$, the control $\bb{u} = [v_x \; v_y]^T$ consists of velocities along $x$ and $y$ axes (in m/s), and $\sigma_x, \sigma_y \in \R$ dictate the strength of noise introduced by the Wiener process $\bb{w} \in \R^2$.
\end{Example}
For control we use \eqref{eq.cbf_qp_controller} with a nominal input of $\bb{u}_0 =  -k[(x - x_g) \; (y - y_g)]^T$ with $k>0$. The input constraints are $|v_x, v_y|\leq v_{max} = 10$, and \eqref{subeq.cbf_constraint} is the S-CBF condition \eqref{eq.scbf_condition}, with $B(\bb{z}) = \frac{x^2 + y^2}{R_c^2}$.
An upper bound on the risk of the system becoming unsafe under the S-CBF-QP controller is then given by \eqref{eq.s_cbf_probability_bounds}. We fixed $\bb{z}_0 = [1/\sqrt{2}, \; 0]^T$ such that $B(\bb{z}_0) = 0.5 = \gamma$ and considered a time horizon of $T=1$ sec at a time-step of $\Delta t = 0.001$ sec. We then simulated the trajectories over $N=100,000$ trials with $\sigma_x, \sigma_y= 0.3v_{max}\cdot \Delta t$, i.e., a strength of $30\%$ of the maximum control input.

The results (shown in Table \ref{tab.stochastic_cbf_empirical_study}) confirmed our hypothesis: the S-CBF constraint may yield a theoretical system risk bound that significantly overestimates the actual fraction of unsafe outcomes.
Despite bounded failure rates of $0.505$ and $0.990$, the S-CBF based controller preserved safety in $100\%$ of the $100,000$ trials ($0$ failures) in both cases over the $T=1$ sec intervals. It is clear from this example that the S-CBF risk bounds may not, and certainly here do not, provide any meaningful guidelines.
\begin{table}[!ht]
    \caption{Stochastic CBF Trials $N=100,000$}\label{tab.stochastic_cbf_empirical_study}\vspace{-3mm}
    \begin{center}
    \begin{tabular}{|c|c|c|c|c|c|c|c|}
        \hline
        Theoretical $\rho$ & Measured $\rho$ & $\alpha$ & $\beta$ & $\gamma$ & $T$ \\ \hline
        0.505 & 0 & 0.1 & 0.01 & 0.50 & 1.0 \\ \hline
        0.990 & 0 & 10.0 & 4.0 & 0.50 & 1.0 \\ \hline
    \end{tabular}
    \end{center}
\end{table}
\vspace{-2mm}

As such, we seek to design a stochastic control framework that bounds the system risk over a finite time interval while bridging the gap between results derived in theory and those observed in practice. We now formally define the problem.
\begin{Problem}\label{prob.problem}
    Consider the stochastic dynamical system of the form \eqref{eq.stochastic_system} and an associated safe set $S$ defined by a twice continuously differentiable, positive semi-definite function $B$ satisfying \eqref{eq.B0_leq_gamma}. Design a feedback controller $\bb{u}_t = k(\bb{x}_t)$ such that under certain conditions $\rho \coloneqq P(\exists t \in [0, T]: \Tilde{\bb{x}}_t \notin S)< \rho_{S\text{-}CBF}$, where $\rho_{S\text{-}CBF}$ is given by \eqref{eq.s_cbf_probability_bounds}, and identify the conditions under which this relation holds.
\end{Problem}

\section{Risk-Aware Control Barrier Function}\label{sec.rb_cbf}

In this section, we propose solving Problem \ref{prob.problem} via a novel class of risk-aware control barrier functions (RA-CBFs). 
First, we require the following.
\begin{Lemma}\label{lem.wiener_crossing_probability}
    Suppose that $w: \R_{\geq 0} \rightarrow \R$ is a standard Wiener process, and $T>0$ and $a>0$ are constants. Then, the probability that $w_t < a$, for all $t \in [0,T]$ is given by
    \begin{align}
        P\left(\sup_{0 \leq t \leq T}w(t) < a \right) = \mathrm{erf}\left(\frac{a}{\sqrt{2T}}\right).
    \end{align}
\end{Lemma}
\begin{proof}
    The proof follows directly from \cite[Section 3]{Blake1973level-crossing}.
\end{proof}


In what follows, we denote the integral of the generator of $\Tilde{\bb{x}}_t$ as
\begin{equation}\label{eq.integrated_cbf_generator}
        I_L(t) \coloneqq \int_0^t\Gamma_B\big(\Tilde{\bb{x}}_s,\bb{u}_s\big)ds,
\end{equation}
which may be included as an integrator state in an augmented system of dimension $n+1$. We now formally introduce the notion of the risk-aware control barrier function.
\begin{Definition}\label{def.risk_aware_cbf}
    Consider a set $S \subset \R^n$ defined by \eqref{eq.safe_set} for a twice continuously differentiable, positive semi-definite function $B$ satisfying \eqref{eq.B0_leq_gamma}. The function $B$ is a \textbf{risk-aware control barrier function} on the set $S$ if there exists a Lipschitz continuous function $\alpha \in \mathcal{K}_\infty$ such that for the system \eqref{eq.stochastic_system} the following holds for all $\bb{x} \in S$,
    \begin{equation}\label{eq.rb_cbf_condition}
        \inf_{\bb{u}\in \mathcal{U}}\Gamma_B(\bb{x},\bb{u}) \leq \alpha\big(h(I_L(t))\big),
    \end{equation}
     where
    \begin{equation}\label{eq.h_rb_cbf}
        h(I_L(t)) = 1 - \gamma - (\sqrt{2}\eta T)\hspace{1pt}\mathrm{erf}^{-1}(1-\rho_d) - I_L(t),
    \end{equation}
    with $I_L(t)$ given by \eqref{eq.integrated_cbf_generator}, $\rho_d \in \left[1-\mathrm{erf}\left(\frac{1 - \gamma}{\sqrt{2}\eta T}\right) ,1\right]$ a design parameter, and
    \begin{equation}\label{eq.eta}
        \eta = \sup_{\bb{x} \in S}\left\|L_\sigma B(\bb{x})\right\|.
    \end{equation}
\end{Definition}
In the following theorem, our main result, we prove that RA-CBFs bound the risk that a system of the form \eqref{eq.stochastic_system} becomes unsafe over a finite time interval.
\begin{Theorem}\label{thm.rb_cbf}
    Let $T>0$, and denote the system risk as $\rho \coloneqq P\left(\exists t\in [0,T]: \Tilde{\bb{x}}_t \notin S \; | \; \Tilde{\bb{x}}_0 \in \mathcal{X}_0\right)$. If $B$ is a risk-aware control barrier function on the set $S$, then,
    \begin{equation}\label{eq.risk_bound_result}
        \rho \leq \rho_d,
    \end{equation}
    where $\rho_d \in \left[1-\mathrm{erf}\left(\frac{1 - \gamma}{\sqrt{2}\eta T}\right) ,1\right]$ is a design parameter with $\eta$ given by \eqref{eq.eta}.
\end{Theorem}
\begin{proof}
    Let $\tau>0$ be the first time of exit of $\bb{x}_t$ from the open set $S$. With $\{\bb{x}_t: t \in [0, \infty)\}$ a strong solution to \eqref{eq.stochastic_system}, we have via It$\hat{\mathrm{o}}$'s Formula \cite[Theorem 4.2.1]{Oksendal2003Stochastic} that $\forall t < \tau$,
    \begin{equation}
        dB(\Tilde{\bb{x}}_t) = \Gamma_B(\Tilde{\bb{x}}_t,\bb{u}_t)dt+ L_\sigma B(\Tilde{\bb{x}}_t)d\bb{w}_t, \nonumber
    \end{equation}
    which leads to the integral equation $B(\Tilde{\bb{x}}_t) = B(\Tilde{\bb{x}}_0) + I_L(t) + I_S(t)$,
    where $I_L(t)$ is a Lebesgue integral defined by \eqref{eq.integrated_cbf_generator} and $I_S(t)$ is a stochastic integral defined by
    \begin{equation}\label{eq.B_ito_integral}
        I_S(t) = \int_0^tL_\sigma B(\Tilde{\bb{x}}_s)d\bb{w}_s. 
    \end{equation}
    While \eqref{eq.integrated_cbf_generator} can be evaluated deterministically, the stochastic integral \eqref{eq.B_ito_integral} is an It$\hat{\mathrm{o}}$ integral \cite[Def. 3.1.6]{Oksendal2003Stochastic} and thus induces a distribution on $B(\Tilde{\bb{x}}_t)$ based on 
    \begin{equation}\nonumber
        I_S(t) \sim \mathcal{N}\left(0, \; \mathbb{E}\left[\left(\int_0^tL_\sigma B(\Tilde{\bb{x}}_s)d\bb{w}_s\right)^2\right]\right).
    \end{equation}
    With $\bb{w}$ the $q$-dim. standard Wiener process, it follows from the $q$-dim. It$\hat{\mathrm{o}}$ isometry (\hspace{-0.2mm}\cite[Lemma 18]{zhang2020wasserstein}  (an extension of the $1$-dim. It$\hat{\mathrm{o}}$ isometry \cite[Lemma 3.1.5]{Oksendal2003Stochastic}) that $\mathbb{E}\left[\left(\int_0^tL_\sigma B(\Tilde{\bb{x}}_s)d\bb{w}_s\right)^2\right] =\int_0^t\left\|L_\sigma B(\Tilde{\bb{x}}_s)\right\|^2ds,$
    and thus that $B(\Tilde{\bb{x}}_t) \sim \mathcal{N}(\mu_B(t), \sigma_B^2(t))$, where $\mu_B(t) = B(\Tilde{\bb{x}}_0) + I_L(t)$ and $\sigma_B^2(t) = \int_0^t\left\|L_\sigma B(\Tilde{\bb{x}}_s)\right\|^2ds$. As such, 
    \begin{align}
        \rho &= 1 - P\left(\sup_{0 \leq t \leq T}B(\Tilde{\bb{x}}_t) < 1 \; | \; B(\Tilde{\bb{x}}_0) \leq \gamma\right). \nonumber
    \end{align}
    Now, let $\bar{I}_S(t) \sim \mathcal{N}(0, \int_0^t\eta^2ds)$ induce a probability distribution on $\bar{B}(\Tilde{\bb{x}}_t)$, i.e., $\bar{B}(\Tilde{\bb{x}}_t) = B(\Tilde{\bb{x}}_0) + I_L(t) + \bar{I}_S(t)$.
    Then, since by \eqref{eq.eta} $\int_0^t\eta^2 ds \geq \sigma_B^2(t)$, for all $t \geq 0$, it follows that
    \begin{equation}\label{eq.rho_bar}
    \begin{aligned}\small
        \rho \leq \bar{\rho} \coloneqq 1 - P\left(\sup_{0 \leq t \leq T}\bar{B}(\Tilde{\bb{x}}_t) < 1 \; | \; B(\Tilde{\bb{x}}_0) \leq \gamma\right).
    \end{aligned}
    \end{equation}
    However, we observe that $\int_0^t\eta^2 ds = \eta^2t$, and thus by Gaussian linearity $\bar{I}_S(t) = \eta\sqrt{t}w(t)$, where $w(t)$ is the $1$-dimensional standard Wiener process, which implies that $\bar{B}(\bb{x}_t) = B_0 + I_L(t) + \eta\sqrt{t}w(t)$. Therefore,
    \begin{align}
        \bar\rho &= 1 - P\left(\sup_{0 \leq t \leq T}w(t) < \frac{1 - \gamma - I_L(t)}{\eta\sqrt{t}} \; | \; B_0 \leq \gamma\right), \nonumber\\
        & \leq 1 - P\left(\sup_{0 \leq t \leq T}w(t) < \frac{1 - \gamma - \bar{I}_L}{\eta\sqrt{T}} \; | \; B_0 \leq \gamma\right), \label{eq.rhobar_bound}
    \end{align}
    where $\bar{I}_L = \sup_{0 \leq t \leq T}I_L(t)$. Thus, from \eqref{eq.rho_bar}, \eqref{eq.rhobar_bound}, and Lemma \ref{lem.wiener_crossing_probability} we have
    \begin{equation}\label{eq.probability_bound}
        \rho \leq \bar\rho \leq 1 - \mathrm{erf}\left(\frac{1 - \gamma - \bar{I}_L}{\sqrt{2}\eta T}\right).
    \end{equation}
    Now, in order for  \eqref{eq.probability_bound} to be true it must hold that $\bar{I}_L \leq 1 - \gamma -(\sqrt{2}\eta T)\mathrm{erf}^{-1}(1-\bar\rho)$,
    a sufficient condition for which is that $I_L(t) \leq 1 - \gamma -(\sqrt{2}\eta T)\mathrm{erf}^{-1}(1-\bar\rho), \quad \forall t \in [0,T]$.
    We then define a set $S_I = \{I_L \in \R \; | \; h(I_L) \geq 0\}$, where $h(I_L) = 1 - \gamma - (\sqrt{2}\eta T)\textrm{erf}^{-1}(1-\bar\rho) - I_L$,
    and observe that if $h$ is a valid CBF for the set $S_I$, i.e., if there exists $\alpha \in \mathcal{K}_\infty$ such that, $\forall I_L \in S_I$ and $\forall t \in [0,T]$, \eqref{eq.rb_cbf_condition} holds then the set $S$ is probabilistically forward-invariant with probability $p=1 - \rho \geq 1-\bar\rho$. Thus, from \eqref{eq.probability_bound} it follows that since $I_L(0)=0$ by definition, $\bar\rho_0 \leq 1 -  \mathrm{erf}\left(\frac{1 - \gamma}{\sqrt{2}\eta T}\right)$ where $\bar{\rho}_0$ is $\bar\rho$ at $t=0$. Therefore, for $h(I_L)\geq 0$ it must hold that $\rho_d \in [\bar\rho_0, 1]$. This completes the proof.
\end{proof}
\begin{Remark}
    Under an RA-CBF controller, the upper bound $\rho$ on the system risk is a function only of the initial condition $\gamma$, the length of the time interval $T$, and the effect of the stochastic noise on $B$, i.e., $\eta$. The function $h$ measures how closely the controller has taken the system to the tolerable risk threshold $\rho$ via actions integrated to form $I_L$.
\end{Remark}
We now present conditions under which the bound on the system risk guaranteed by Theorem \ref{thm.rb_cbf} is strictly less than the bound guaranteed under the S-CBF control framework.

\begin{Theorem}\label{thm.tighter_risk_bound}
    Let the premises of Theorem \ref{thm.stochastic_cbf_bounds} hold, and let $\rho_d$ be as defined in Theorem \ref{thm.rb_cbf}. If $B$ is a risk-aware control barrier function, then
    \begin{equation}
        \min_{\rho_d \in \mathcal{R}} \rho_d < \rho_{S\text{-}CBF}
    \end{equation}
    whenever
    \begin{equation}\label{eq.eta_condition}
        \eta < \frac{1 - \gamma}{\sqrt{2}T\mathrm{erf}^{-1}(1-\gamma)},
    \end{equation}
    where $\eta$ is given by \eqref{eq.eta} and $\mathcal{R} = [1-\mathrm{erf}\left(\frac{1 - \gamma}{\sqrt{2}\eta T}\right),1]$.
\end{Theorem}
\begin{proof}
    The proof follows immediately from the observation in Remark \ref{rmk.s_cbf_bounds}, i.e., that $\rho_{S\text{-}CBF} \geq \gamma$, $\forall \alpha, \beta, T \geq 0$. Then, from Theorem \ref{thm.rb_cbf}, $\min \rho_d < \rho_{S\text{-}CBF}$ whenever $1 - \mathrm{erf}\left(\frac{1-\gamma}{\sqrt{2}\eta T}\right) < \gamma$. By rearranging terms, we recover \eqref{eq.eta_condition}.
\end{proof}
This result provides guidelines as to when a RA-CBF controller would predict lower levels of risk than a S-CBF controller, or vice versa. In the robot motion problem from Section \ref{subsubsec.problem}, with dynamics \eqref{eq.2d_single_integrator_model} and barrier function $B(\bb{z}) = \frac{x^2 + y^2}{R_c^2}$ we have $\eta \approx 0.009$. As such, $\min \rho_d \geq \rho_{S\text{-}CBF}$ over the $T=1$ sec time interval would have required either $\gamma < 1e$-$15$ given $\sigma_x,\sigma_y$ or $\sigma_x,\sigma_y \approx 50v_{max}\cdot \Delta t$ given $\gamma=0.5$, both of which are unrealistic for the problem.

When $\eta$ given by \eqref{eq.eta} is large, however, the allowable risk specifications using a RA-CBF (based on $\min \rho_d$) may not be acceptable.
In this case, it may be more useful to design the controller to remain inside a smaller sub-level set $S_\mu = \{\bb{x} \in \R^n: \; 0 \leq B(\bb{x}) < \mu\}$, or to derive a total risk of the system becoming unsafe by cascading sets $S_{\mu_1},\hdots,S_{\mu_k}$, as shown in the following result.

\begin{Theorem}\label{thm.risk_cascading_sets}
    Suppose that the premises of Theorem \ref{thm.rb_cbf} hold. Consider a sequence $\mu_0,\hdots,\mu_k$ such that $\gamma = \mu_0 < \mu_1 < \hdots < \mu_k = 1$ with sub-level sets $S_{\mu_i} = \{\bb{x} \in \R^n: \; 0 \leq B(\bb{x}) < \mu_i\} \subseteq S$, $\forall i \in \{1,\hdots,k\}$, each of which has $\eta_i$ defined by \eqref{eq.eta} over $S_{\mu_i}$. If $B$ is a RA-CBF on each set $S_{\mu_i}$, then
    $\rho \leq \rho_d$, where $\rho_d$ is a design parameter bounded by
    \begin{equation}\label{eq.tighter_cascaded_rho_bound}
        \prod_{i=1}^k\left(1 - \mathrm{erf}\left(\frac{\mu_{i} - \mu_{i-1}}{\sqrt{2}T\eta_{i}}\right)\right) \leq \rho_d \leq 1.
    \end{equation}
\end{Theorem}
\begin{proof}
First, observe that by Definition \ref{def.risk_aware_cbf} the function $B$ is a RA-CBF on the set $S_{\mu_i}$ if \eqref{eq.rb_cbf_condition} holds for all $\bb{x} \in S_{\mu_i}$, where the $1-\gamma$ term in \eqref{eq.h_rb_cbf} is replaced by $\mu_i - \gamma$. Let $\rho_{\mu_i} \coloneqq P\left(\exists t\in [0,T]: \Tilde{\bb{x}}_t \notin S_{\mu_i} \; | \; \Tilde{\bb{x}}_0 \in S_{\mu_i} \setminus S_{\mu_{i-1}}\right)$. Then, with $B$ a RA-CBF on $S_{\mu_i}$, it follows from Theorem \ref{thm.rb_cbf} that 
\begin{equation}\label{eq.rho_mu_k}
    \rho_{\mu_i} \leq 1 - \mathrm{erf}\left(\frac{\mu_{i} - \mu_{i-1}}{\sqrt{2}T\eta_{i}}\right),
\end{equation}
where $\eta_{\mu_i}$ is defined by \eqref{eq.eta} over the set $S_{\mu_i}$. By \eqref{eq.B0_leq_gamma} and Bayes' rule, we then obtain that $\rho \leq \prod_{i=1}^k\rho_{\mu_i}$ and thus by \eqref{eq.rho_mu_k} we recover \eqref{eq.tighter_cascaded_rho_bound}.
\end{proof}

The bound in \eqref{eq.rho_mu_k} is particularly useful when $\eta_1 < \hdots < \eta_k$, as this is the best reduction in the conservatism in using $\eta$ over all $S$. The number of partitions $k$ is a design choice, and should be adjusted according to the desired system risk and each $\eta_i$. For control design, the RA-CBF condition \eqref{eq.rb_cbf_condition} must be satisfied on each $S_{\mu_i}$ with a choice of $\rho_{d_i} \geq \rho_{\mu_i}$.

\section{Numerical Case Studies}\label{sec.case_studies}

In this section, we highlight the efficacy of our RA-CBF controller in solving two illustrative examples: the robot problem from Section \ref{subsubsec.problem}, and a highway merging problem.

\subsection{Single-Integrator Robot}
The problem setup is identical to that in Section \ref{subsubsec.problem}, with the robot's dynamics given by \eqref{eq.2d_single_integrator_model} and its controller of the form \eqref{eq.cbf_qp_controller} with RA-CBF condition \eqref{eq.rb_cbf_condition}. 
The results were a striking departure from the S-CBF based controller. When an upper bound on system risk was set to $0.505$ to match the S-CBF trial, a fraction of $0.458$ of the trials violated the safety condition, as shown in Table \ref{tab.rb_cbf_empirical_study}.
\begin{table}[!ht]
    \centering
    \caption{Risk-Aware CBF Trials $N=100,000$}\label{tab.rb_cbf_empirical_study}
    \begin{tabular}{|c|c|c|c|c|c|c|c|}
        \hline
        Predicted $\rho$ & Measured $\rho$ & $\gamma$ & $\eta$ \\ \hline
        0.010 & 10$^{-4}$ & 0.50 & 0.006 \\ \hline
        0.505 & 0.458 & 0.50 & 0.006 \\ \hline
    \end{tabular}
\end{table}
\vspace{-2mm}
When the RA-CBF controller was used at a maximum system risk of $\rho = 0.01$, however, not only did the measured $\rho$ satisfy this bound ($1e$-$4$), but the system trajectories took more aggressive actions toward the boundary of the safe set than the S-CBF controller even when its risk level was set to $\rho_{S\text{-}CBF} = 0.505$, as shown in Figure \ref{fig.barrier_function_comparison}.
\begin{figure}[!ht]
    \centering
        \includegraphics[width=1\columnwidth,clip]{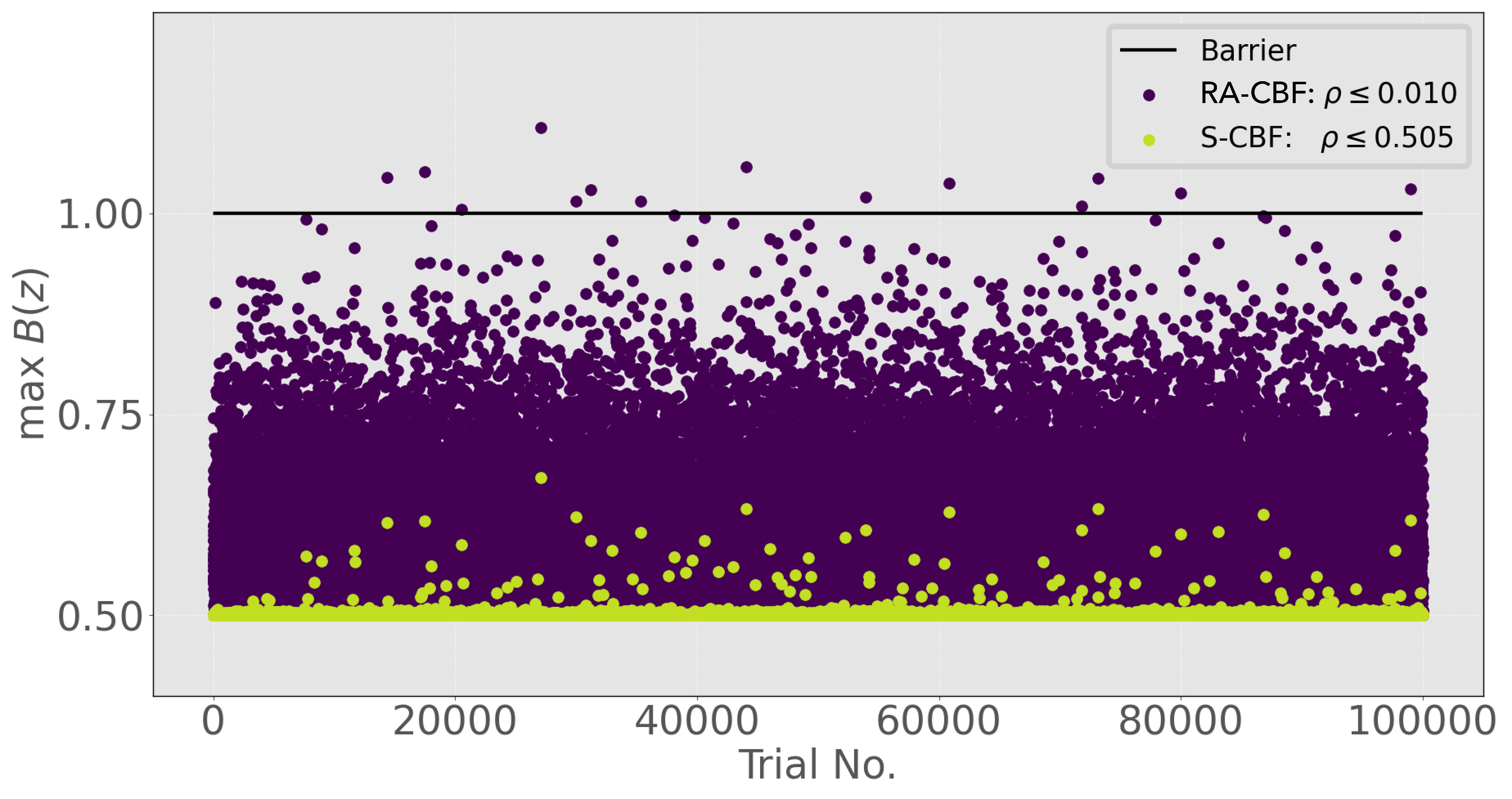}
    \caption{Maximum barrier function values ($\max_{0 \leq t \leq T}B(\bb{z}_t)$) over each trial for RA-CBF (resp. S-CBF) with system risk bounded by $\rho \leq 0.01$ (resp. $\rho_{S\text{-}CBF} \leq 0.505$).}\label{fig.barrier_function_comparison}
    \vspace{-2mm}
\end{figure}

\subsection{Highway Merging}
Let $\mathcal{I}$ be an inertial frame with an origin point $s_0$. Consider a collection of automobiles $\mathcal{A}$, a subset of which travel on a two-lane highway near an on-ramp (i.e., $\mathcal{A}_H \subset \mathcal{A}$), and the remainder of which seek to merge onto the highway via the on-ramp (i.e., $\mathcal{A}_M \subset \mathcal{A}$). Suppose that the dynamics of vehicle $i \in \mathcal{A}$ obey a stochastic bicycle model of the form \eqref{eq.stochastic_system} whose deterministic component is described by \cite[Ch. 2]{Rajamani2012VDC} (omitted due to space) and used to model cars in \cite{Black2022ffcbf}.
The stochastic term is $\sigma_i(\bb{z}_{i}) = \bb{\sigma}^T\mathbf{I}_{5 \times 5}$ with $\bb{\sigma} = [0 \; 0 \; 0 \; \sigma_{a} \; \sigma_\omega]^T$. The state is $\bb{z}_i = [x_i \; y_i \; \psi_i \; v_{i} \; \beta_i]^T$, where $x_i$ and $y_i$ denote the longitudinal and lateral positions (in m) of the center of gravity (c.g.) of vehicle $i$ with respect to $s_0$, $\psi_i$ is the orientation (in rad) of its body-fixed frame, $\mathcal{B}_i$, with respect to $\mathcal{I}$, $v_{i}$ is the velocity (in m/s) of the rear wheel with respect to $s_0$, and $\beta_i$ is the slip angle\footnote{
$\beta_i$ is related to the steering angle $\delta_i$ via $\tan{\beta_i} = \frac{l_r}{l_r+l_f}\tan{\delta_i}$.} (assume $|\beta_i|<\frac{\pi}{2}$) of the c.g. relative to $\mathcal{B}_i$ (in rad). The front and rear wheelbases are $l_f$ and $l_r$. See \cite[Figure 1]{Black2022ffcbf} for a model diagram. The control is $\bb{u}_i = [a_{i} \; \omega_i]^T$, with $a_{i}$ the linear acceleration of the rear wheel (in m/s$^2$) and $\omega_i$ the angular velocity (in rad/s) of $\beta_i$. The vector $\bb{w} \in \R^5$ is the $5$D standard Wiener process. 
\begin{figure*}[!ht]
\centering
\includegraphics[clip,width=1\textwidth]{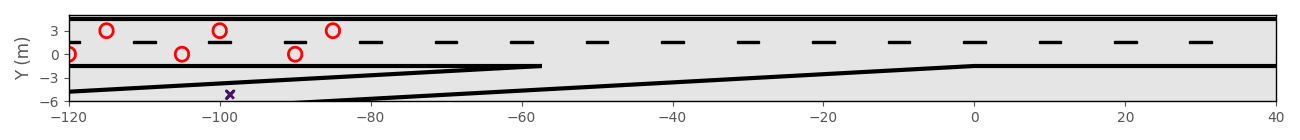}\\
\centering
\includegraphics[clip,width=1\textwidth]{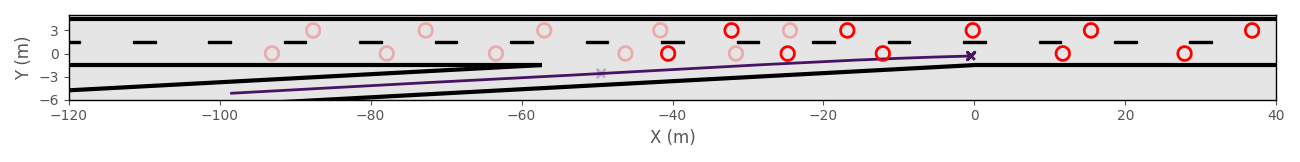}
\caption{Snapshots at $t=0.0$s (a) and $t=4.0$s (with $t=2.0s$ translucent) (b) of one trial from the empirical study on the RA-CBF-QP controller in the highway merging scenario. Traffic flows left to right, the ego vehicle is a blue X, and highway vehicles are red circles.}\label{fig.merging_example}
\vspace{-2mm}
\end{figure*}
\begin{figure}[!ht]
    \centering
        \includegraphics[width=1\columnwidth,clip]{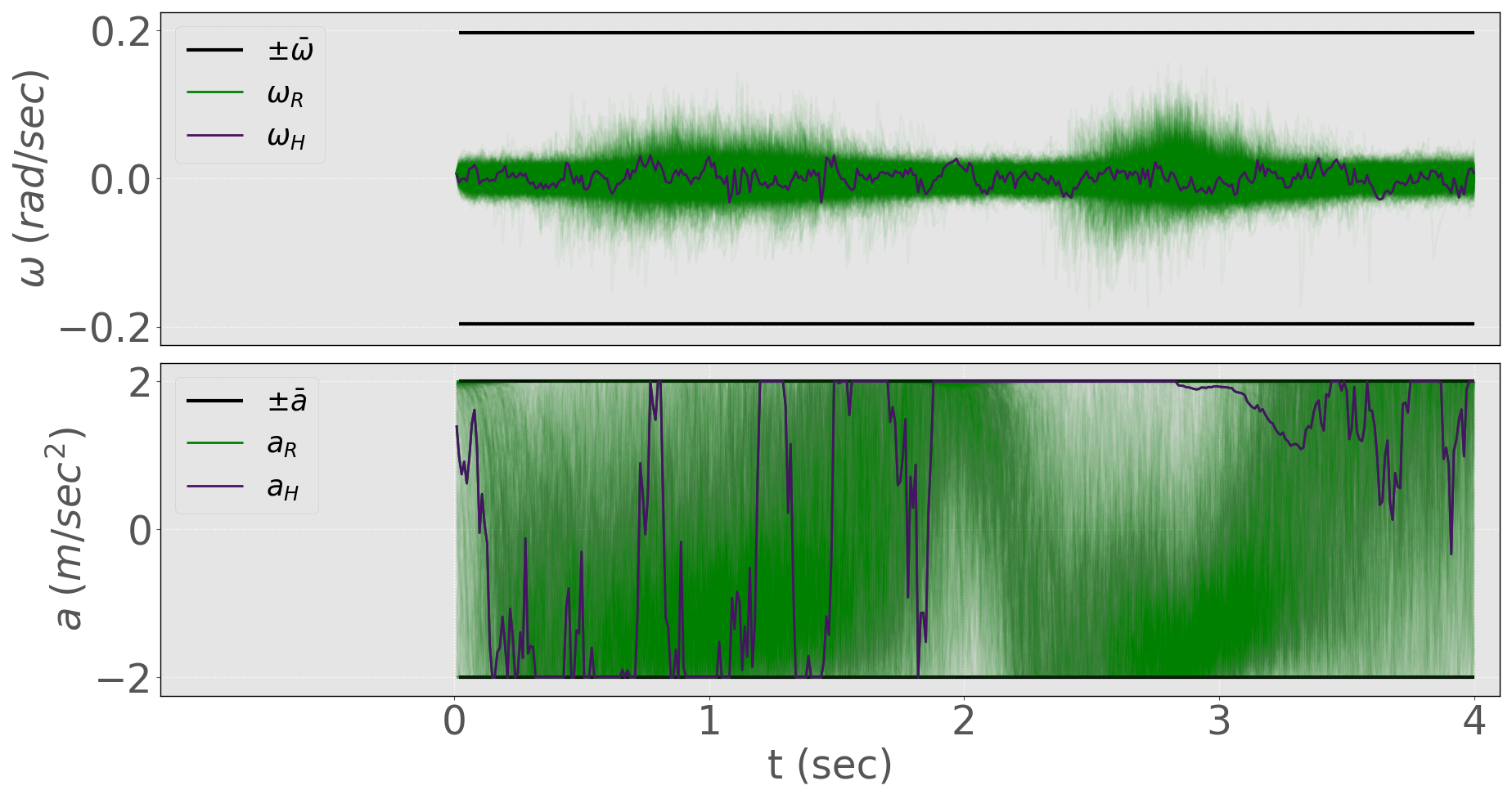}
    \caption{Ego vehicle control inputs from both the highlighted trial in Figure \ref{fig.merging_example} (subscript H) and remaining trials (subscript R).}\label{fig.control_inputs}
\end{figure}
\begin{figure}[!ht]
    \centering
        \includegraphics[width=1\columnwidth,clip]{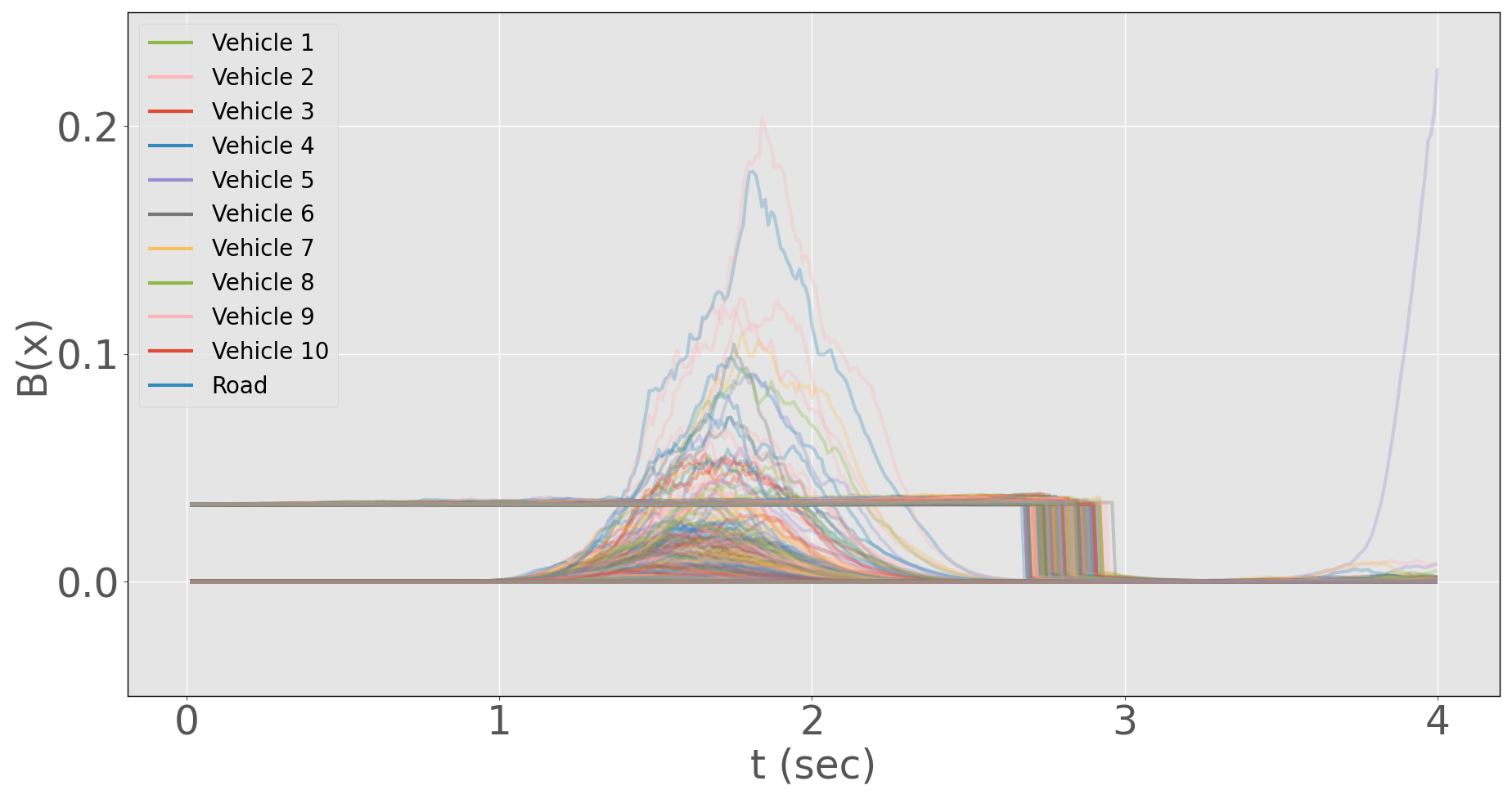}
    \caption{RA-CBF trajectories over 1000 highway merging trials.}\label{fig.cbf_trajectories}
\end{figure}

We simulated 1000 trials of a $T=4$s highway merging scenario with 11 vehicles, where $\mathcal{A}_M = \{0\}$ was the ego vehicle and $\mathcal{A}_H = \{1,\hdots,10\}$. Highway vehicles $i \in \mathcal{A}_H$ were initialized $15$m apart in the $x$ direction, and distributed evenly between lanes 1 ($y=0$) and 2 ($y=3$). Their initial velocities were distributed according to $v_{i,0} \sim U[29, 31]$. The ego vehicle was initialized $98.75$m down the on-ramp with an initial velocity $v_{e,0} \sim U[24, 26]$ such that (deterministically) under its nominal acceleration policy it would collide directly with vehicle 2. The noise components were $\sigma_{a} = A_{drag}\Delta t$ and $\sigma_\omega = \sigma_{a} \frac{\bar\omega}{\bar{a}}$, where $\bar\omega=\frac{\pi}{16}$ (rad/s) and $\bar{a}=2.0$ (m/s$^2$) define the input constraints $a_i \in [-\bar a, \bar a]$ and $\omega_i \in [-\bar \omega, \bar \omega]$, and $A_{drag}=0.1 + 5\bar{v} + 0.25\bar{v}^2$, with $\bar{v}=35$ (m/s), such that the noise at one standard deviation represents the acceleration due to aerodynamic drag \cite{Ames2017CBFs} traveling at $35$m/s. 
The ego vehicle was controlled using \eqref{eq.cbf_qp_controller} with 11 RA-CBF constraints corresponding to the occupied sub-level set $S_{\mu_i}$, where we selected $\mu_i = i / 5$ for $i \in \{1,\hdots,5\}$. The 10 ego collision avoidance constraints were encoded via $B_{ei}(\bb{z}_e, \bb{z}_i) = e^{-h_{ei}(\bb{z}_e, \bb{z}_i)}$, where $h_{ei}(\bb{z}_e,\bb{z}_i)$ is the relaxed future-focused CBF (rff-CBF) (introduced for collision avoidance in \cite{Black2022ffcbf}) with $\gamma(h_0) = 0.1h_0$. The road constraints were encoded with a rff-CBF of the form $B_r(\bb{z}_e) = e^{-h_r(\bb{z}_e)}$ for $h_r(\bb{z}_e) = h_{r,0}(\bb{z}_e, 0) + h_{r,0}(\bb{z}_e, 1)$,  where
\begin{equation}
    h_{r,0}(\bb{z}_e, \tau) = -\left(x_e\tan(\theta) + \frac{w_l}{2\cos(\theta)} - (y_e + \dot{y}_e \tau - y_l)\right)^2 \nonumber
\end{equation}
with $\theta$ the road angle with respect to the $x$-axis, $w_l$ the width of a lane, and $y_l$ the lane center. For all RA-CBFs, the corresponding $\eta_i$ values were determined numerically by simulating 1000 trials and taking $\eta_i=\max_{\bb{x} \in S_{\mu_i}}\left\|L_\sigma B(\bb{x})\right\|$ over all trials and all time. The resulting $\eta_i$ values are provided in Table \ref{tab.eta_values}. For the on-ramp, the angle of attack was $\theta=3^\circ$ ($\theta=0^\circ$ for highway lanes). The ego nominal control $\bb{u}_{0}$ was the LQR law detailed in \cite[Appendix 1]{Black2022ffcbf} based on the desired lane and velocity ($v_d = 30$m/s).
\begin{table}[!ht]
    \centering
    \caption{$\eta_i$ values derived empirically}\label{tab.eta_values}
    \begin{tabular}{|c|c|c|c|c|c|c|c|c|c|c|c|c|c|c|c|}
        \hline
        CBF & $\eta_1$ & $\eta_2$ & $\eta_3$ & $\eta_4$ & $\eta_5$ \\ \hline
        Road & 0.012 & 0.025 & 0.035 & 0.046 & 0.067 \\ \hline
        Collision & 0.018 & 0.031 & 0.049 & 0.063 & 0.076 \\ \hline
    \end{tabular}
    \vspace{-2mm}
\end{table}
For naturalistic driving behavior, we used the intelligent driver model (IDM)\cite{Treiber2000idm} to compute acceleration inputs $a_{i}$ of the highway vehicles $i \in \mathcal{A}_H$. For varying driver aggression, we randomized the vehicles' desired time gaps in the IDM according to $\tau \sim U[0.25, 0.75]$. Their steering inputs $\omega_i$ were computed using LQR based on the desired heading ($\psi_d = 0$).

Based on $\eta_i$ from Table \ref{tab.eta_values}, a simulation length of $T=4$ sec, known $\gamma$ for all $B(\bb{z}_0)$, the minimum specifiable risks associated with leaving each sub-level set $S_{\mu_i}$ for road safety and collision avoidance are provided in Table \ref{tab.risk_values}.
\begin{table}[!ht]
    \centering
    \caption{$\rho_i$ values for sub-level sets $S_{\mu_i}$}\label{tab.risk_values}
    \begin{tabular}{|c|c|c|c|c|c|c|c|c|c|c|c|c|c|c|c|}
        \hline
        CBF & $\rho_1$ & $\rho_2$ & $\rho_3$ & $\rho_4$ & $\rho_5$ \\ \hline
        Road & $8.58$e-$4$ & 0.046 & 0.153 & 0.277 & 0.456 \\ \hline
        Collision & 0.026 & 0.107 & 0.308 & 0.427 & 0.511 \\ \hline
    \end{tabular}
    \vspace{-3mm}
\end{table}
We chose the $\rho_{d,i}$ values provided in Table \ref{tab.rho_d_values} such that the probability of remaining safe with respect to the road is $0.99999$, the probability of remaining safe with respect to all 10 highway vehicles combined is $0.991$, and thus the total probability of safety is $p\geq 0.99$, which yields $\rho \leq 0.01$.
\begin{table}[!ht]
    \centering
    \caption{Specified risk bounds $\rho_{d,i}$ for sub-level sets $S_{\mu_i}$}\label{tab.rho_d_values}
    \begin{tabular}{|c|c|c|c|c|c|c|c|c|c|c|c|c|c|c|c|}
        \hline
        CBF & $\rho_{d,1}$ & $\rho_{d,2}$ & $\rho_{d,3}$ & $\rho_{d,4}$ & $\rho_{d,5}$ \\ \hline
        Road & 0.001 & 0.1 & 0.25 & 0.5 & 0.6 \\ \hline
        Collision & 0.05 & 0.15 & 0.4 & 0.5 & 0.6 \\ \hline
    \end{tabular}
    \vspace{0mm}
\end{table}


Over 1000 simulated trials, the RA-CBF based controller safely merged 1000 times, satisfying the risk bound of $\rho \leq 0.01$. Figure \ref{fig.merging_example} highlights one of these safe merges in which the ego vehicle merges behind vehicle 2, where the applied control inputs are shown in Figure \ref{fig.control_inputs}. In this study, we observed that in all 1000 trials the ego vehicle merged behind vehicle 2. In another study, in which a risk of $\rho \leq 0.12$ was specified, we observed that the ego vehicle safely merged in 914 of the 1000 trials ($p=0.914$). Interestingly, of these 914 safe trials, the ego vehicle merged behind vehicle 2 at a rate of $0.749$ and merged ahead of it the remaining $0.251$ fraction of safe trials, an indicator of the willingness of the ego vehicle in the second study to take on additional risk.

\section{Conclusion}\label{sec.conclusion}
In this paper, we proposed a new class of RA-CBFs for stochastic safety-critical systems. We introduced a new CBF condition for a class of stochastic, nonlinear, control-affine systems and proved that its use for control synthesis guarantees an upper bound on the risk that the system becomes unsafe over a finite time interval. We then derived conditions under which our RA-CBF controller results in a smaller system risk than existing methods, and conducted a direct comparative study on a mobile robot example. We demonstrated our control strategy under a $99.1\%$ safety guarantee on an autonomous vehicle highway merging problem in the midst of dense traffic. In the future, we will consider measurement noise and investigate applications of our control framework to recovery problems in the context of safe, predictive control.

\bibliographystyle{IEEEtran}
\bibliography{library}

\end{document}